\documentclass[11pt,a4paper]{article}
\usepackage[english]{babel}
\usepackage[latin1]{inputenc}
\usepackage{amsmath, amsthm, amssymb}
\usepackage[pdftex]{graphicx}
\usepackage{amssymb}
\usepackage{latexsym}
\usepackage{amstext}
\usepackage{epstopdf}
\usepackage{hyperref}
\usepackage{enumitem}
\usepackage{pdflscape}
\usepackage{tikz}
\usetikzlibrary{positioning,arrows}
\usetikzlibrary{decorations.pathmorphing}
\usetikzlibrary{decorations.markings}

\newtheorem{theorem}{Theorem}[section]

\newtheorem{proposition}[theorem]{Proposition}
\newtheorem{remark}[theorem]{Remark}

\newcommand{\adj}{\text{\normalfont ad}}

\newcommand{\as}{\alpha}

\newcommand{\au}{\mathbf a_1}

\newcommand{\ba}{\begin{array}}

\newcommand{\bb}{\beta}
\newcommand{\bc}{\begin{center}}
\newcommand{\bdo}{{\mathfrak B}_{12}}
\newcommand{\be}{\begin{equation}}
\newcommand{\bea}{\begin{equation}\begin{array}}
\newcommand{\beas}{\begin{equation*}\begin{array}}

\newcommand{\bef}{\begin{flalign}}
\newcommand{\befs}{\begin{flalign*}}
\newcommand{\ben}{\begin{enumerate}}
\newcommand{\benal}[1]{\begin{enumerate}[label={#1}\alph*)]}
\newcommand{\benar}[1]{\begin{enumerate}[label={#1}\arabic*)]}
\newcommand{\benro}[1]{\begin{enumerate}[label={#1}\roman*)]}
\newcommand{\benRo}[1]{\begin{enumerate}[label={#1}\Roman*)]}
\newcommand{\bes}{\begin{equation*}}
\newcommand{\bit}{\begin{itemize}}

\newcommand{\Bp}{{\mathfrak B}^+}

\newcommand{\de}[1]{\lfloor #1\rfloor}
\newcommand{\der}{\adj(R)}
\newcommand{\dotto}{\mathbf{d_8}}

\newcommand{\ea}{\end{array}}
\newcommand{\ec}{\end{center}}

\newcommand{\ee}{\end{equation}}
\newcommand{\eea}{\end{array}\end{equation}}
\newcommand{\eeas}{\end{array}\end{equation*}}
\newcommand{\eef}{\end{flalign}}
\newcommand{\eefs}{\end{flalign*}}
\newcommand{\een}{\end{enumerate}}
\newcommand{\ees}{\end{equation*}}

\newcommand{\eit}{\end{itemize}}

\newcommand{\eno}{\mathbf{e_9}}
\newcommand{\eo}{\mathbf{e_8}}

\newcommand{\ep}{\varepsilon}
\newcommand{\eps}{\varepsilon}

\newcommand{\esquare}{\hfill $\square$}

\newcommand{\ggu}{\textbf{${\mathfrak g}_{\mathsf u}$}}
\newcommand{\ggue}{\textbf{${\mathfrak g}^e_{\mathsf u}$}}
\newcommand{\gguo}{\textbf{${\mathfrak g}_{\mathsf u 0}$}}
\newcommand{\ggup}{\textbf{${\mathfrak g}^+_{\mathsf u}$}}
\newcommand{\ggupe}{\textbf{${\mathfrak g}^{+e}_{\mathsf u}$}}
\newcommand{\gguu}{\textbf{${\mathfrak g}_{\mathsf u 1}$}}

\newcommand{\gh}{\gamma}

\newcommand{\ical}{\mathcal I}
\newcommand{\impl}{\ \Rightarrow\ }

\newcommand{\Kass}{\mathbb K}

\newcommand{\lam}{\lambda}
\newcommand{\Ll}{\mathbb L}

\newcommand{\lore}{\mathcal{M}}

\newcommand{\nn}{\mathbb N}

\newcommand{\Perm}{\mathbb P}

\newcommand{\poi}{\mathcal{P}}
\newcommand{\poig}{\mathfrak{P}}

\newcommand{\pt}{\tilde p}

\newcommand{\rmu}{\as_{\text{-}1}}

\newcommand{\rop}{\as_{0'}}
\newcommand{\ropp}{\as_{0''}}

\newcommand{\sggu}{\textbf{$\mathfrak {sg}_{\mathsf u}$}}

\newcommand{\spin}{\mathbf{su(2)^{spin}}}

\newcommand{\sref}[1]{{\bf\ref{#1}}}

\newcommand{\tp}{{\scriptstyle \, \otimes\,}}

\newcommand{\usgu}{{\mathfrak U}_{\sggu}}

\newcommand{\zz}{\mathbb Z}

\numberwithin{equation}{section}
\begin{document}

\begin{titlepage}

\vskip 2.0 cm
\begin{center}  {\Huge{\bf Space, Matter and Interactions\\\vskip 0.2 cm in a Quantum Early Universe}} \\\vskip 0.5 cm {\huge{\bf Part II : Superalgebras and Vertex Algebras}}

\vskip 2.5 cm

{\Large{\bf Piero Truini$^{1,2}$},  {\bf Alessio Marrani$^{3,4}$},\\ {\bf Michael Rios$^{2}$}, {\bf Klee Irwin$^{2}$}}

\vskip 1.0 cm

$^1${\sl INFN, sez. di Genova, via Dodecaneso 33, I-16146 Genova, Italy\\
	\texttt{piero.truini@ge.infn.it}}\\

$^2${\sl
QGR, 101 S. Topanga Canyon Rd., 1159 Los Angeles,CA 90290, USA}\\

\vskip 0.5
cm

\vskip 0.5 cm

$^3${\sl Centro Studi e Ricerche Enrico Fermi, via Panisperna 89A,
I-00184, Roma, Italy}\\

\vskip 0.5 cm

$^4${\sl Dipartimento di Fisica
e Astronomia Galileo Galilei, Universit\`a di Padova,\\and INFN, sezione
di
Padova, Via Marzolo 8, I-35131 Padova, Italy\\
\texttt{jazzphyzz@gmail.com}}\\

\vskip 0.5 cm

\vskip 0.5 cm

 \end{center}

 \vskip 4.0 cm

\begin{%
abstract}

In our investigation on quantum gravity, we introduce an infinite dimensional complex Lie algebra $\ggu$ that extends $\eno$. It is defined through a symmetric Cartan matrix of a rank 12 Borcherds algebra. We turn $\ggu$ into a Lie superalgebra $\sggu$ with no superpartners, in order to comply with the Pauli exclusion principle. There is a natural action of the Poincar\'e group on $\sggu$, which is an automorphism in the massive sector. We introduce a mechanism for scattering that includes decays as particular {\it resonant scattering}. Finally, we complete the model by merging the local $\sggu$ into a vertex-type algebra.
 \end{abstract}
\vspace{24pt} \end{titlepage}

%{\parskip -6pt }

%\parskip 9pt
\newpage \tableofcontents \newpage

%%%%%%%%%%%%%%%%%%%%%%%%%%%%%%%%%%%%%%%%%%%%%%%%%%%%%%%%%%%%%%%%%%%%%%%%%%%%%%%%%%%%%%%%%%%%%%%%%%%%%%%%%%%%%%%%

\section{Introduction}

This is the second of two papers - see also \cite{mym1} - describing an
algebraic model of quantum gravity.\newline
In the first paper we have described the basic principles of our model and
we have investigated the mathematical structures that may suit our purpose.
In particular, we have focused on rank-12 infinite dimensional Kac-Moody,
\cite{kac}, and Borcherds algebras, \cite{borc1} \cite{borc2}, and we have
given physical and mathematical reasons why the latter are preferable.%
\newline

In our model for the expansion of quantum early Universe, \cite{mym1} \cite%
{pt2}, the need for an infinite dimensional Lie algebra stems from the
unlimited number of possible 4-momenta, but at each fixed cosmological time
the number of generators and roots involved is \textit{finite}. There is a
known algorithm of Lie algebra theory that allows to determine the structure
constants among a finite number of generators of a Borcherds algebra \cite%
{borc1} \cite{borc2}. Let us grade the commutators by \textit{levels}, by
saying that the commutators involving $n$ simple roots have level $n-1$. A
consistent set of structure constants is calculable \textit{level by level},
and once the structure constants are calculated at level $n$, they will not
be affected by the calculation at any level $m>n$. There are computer
programs that apply this algorithm and give the explicit structure constants
level by level, see for instance the package \textit{LieRing} of GAP,
developed by S. Cical\`{o} and W. A. de Graaf \cite{gap}.\newline

However, for the sake of simplicity, in \cite{mym1} we have chosen to deal
with a simpler Lie algebra, $\ggu,$ that extends $\eo$ and $\eno$. In the
present paper, we will start invertigating a physical model for quantum
gravity based on this particular rank-12 algebra $\ggu$. We will start by
focussing on local aspects of the algebraic model: in Sec. \sref{s:ggu}, we
recall $\ggu$, which is then turned into a Lie superalgebra $\sggu$ in Sec. %
\sref{s:sggu}. Sec. \sref{s:scatt} will then discuss interactions,
scattering processes and decays, whereas the role of the Poincar\'{e} group
is analyzed in Sec. \sref{s:Poin}. Finally, in Secs. \sref{s:stat} and %
\sref{s:vert} we will define the quantum states, and then we will merge the
algebra $\sggu$ into a vertex-type algebra, representing the quantum early
Universe with its expanding spacetime.

\section{The Lie algebra $\ggu$\label{s:ggu}}

We start and consider $\Bp$, the Lie subalgebra of the rank-12 Borcherds
algebra $\bdo$ introduced in \cite{mym1} and generated by the Chevalley
generators corresponding to positive roots. A further simplification will
then give rise to $\ggu$, the Lie algebra that acts \textit{locally} on the
quantum state of the Universe \cite{mym1}.

We recall from \cite{mym1} that the generalized Cartan matrix for the
Borcherds algebra $\bdo$, with simple roots denoted by $\rmu,\ropp,\rop,\as%
_{0},...,\as_{8}$, is%
\begin{equation}
\left(
\begin{array}{cccccccccccc}
-1 & -1 & -1 & -1 & 0 & 0 & 0 & 0 & 0 & 0 & 0 & 0 \\
-1 & 0 & -1 & -1 & 0 & 0 & 0 & 0 & 0 & 0 & 0 & 0 \\
-1 & -1 & 0 & -1 & 0 & 0 & 0 & 0 & 0 & 0 & 0 & 0 \\
-1 & -1 & -1 & 2 & -1 & 0 & 0 & 0 & 0 & 0 & 0 & 0 \\
0 & 0 & 0 & -1 & 2 & -1 & 0 & 0 & 0 & 0 & 0 & 0 \\
0 & 0 & 0 & 0 & -1 & 2 & -1 & 0 & 0 & 0 & 0 & 0 \\
0 & 0 & 0 & 0 & 0 & -1 & 2 & -1 & 0 & 0 & 0 & 0 \\
0 & 0 & 0 & 0 & 0 & 0 & -1 & 2 & -1 & 0 & 0 & 0 \\
0 & 0 & 0 & 0 & 0 & 0 & 0 & -1 & 2 & -1 & -1 & 0 \\
0 & 0 & 0 & 0 & 0 & 0 & 0 & 0 & -1 & 2 & 0 & 0 \\
0 & 0 & 0 & 0 & 0 & 0 & 0 & 0 & -1 & 0 & 2 & -1 \\
0 & 0 & 0 & 0 & 0 & 0 & 0 & 0 & 0 & 0 & -1 & 2%
\end{array}%
\right)  \label{cm12}
\end{equation}
By defining%
\begin{equation}
\delta :=\as_{0}+2\as_{1}+3\as_{2}+4\as_{3}+5\as_{4}+6\as_{5}+3\as_{6}+4\as%
_{7}+2\as_{8},
\end{equation}
the 4-momentum vector can be written as%
\begin{equation}
p:=E_{p}\rmu+p_{x}(\ropp-\rmu)+p_{y}(\rop-\rmu)+p_{z}(\delta -\rmu).
\label{fmom}
\end{equation}

Then, we restrict to the subalgebra $\Bp$ of $\bdo$, namely to positive
roots $r=\sum_{\ical}\lam_{i}\as_{i}$, $\ical:=\{-1,0^{\prime \prime
},0^{\prime },0,...,8\}$, with $\lam_{i}\in \nn\cup \{0\}$. Consequently,
the 4-momentum (\ref{fmom}) becomes%
\begin{equation}
p=(E_{p},p_{x},p_{y},p_{z})=(\lam_{-1}+\lam_{0\prime \prime }+\lam_{0\prime
}+\lam_{0},\lam_{0\prime \prime },\lam_{0\prime },\lam_{0})  \label{palfa}
\end{equation}
with $\lam_{-1},\lam_{0^{\prime \prime }},\lam_{0^{\prime }},\lam_{0}\geq 0$%
, implying%
\begin{equation}
m^{2}:=-p^{2}\geq 0,
\end{equation}
namely $p$ either lightlike or timelike. In particular ($i,j\in \left\{
-1,0^{\prime \prime },0^{\prime },0\right\} $),%
\bea{ll}
p^{2} &=-\left( \lambda _{-1}^{2}+2\lambda _{-1}\sum_{i\neq -1}\lambda
_{i}+\sum_{i\neq j,~i,j\neq -1}\lambda _{i}\lambda _{j}\right)\\
& \left\{
\begin{array}{ll}
=0 & \text{if~}\lambda _{-1}=0~\text{and~at~most~one~}\lambda _{i}\neq
0,~i\neq -1\text{,} \\
=-1 & \text{if~}\lambda _{-1}=1~\text{and~all~}\lambda _{i}=0,~i\neq -1\text{%
,} \\
\leqslant -2 & \text{otherwise}.%
\end{array}%
\right.\\
\eea
As in \cite{mym1}, we write a root $r=\sum_{\ical}\lam_{i}\as_{i}$ as%
\begin{equation}
r=\alpha +p,
\end{equation}%
with%
\begin{eqnarray}
\Phi _{8} &\ni &\alpha =(\lam_{1}-2\lam_{0})\as_{1}+(\lam_{2}-3\lam_{0})\as%
_{2}+(\lam_{3}-4\lam_{0})\as_{3}\newline
\notag \\
&&+(\lam_{4}-5\lam_{0})\as_{4}+(\lam_{5}-6\lam_{0})\as_{5}+(\lam_{6}-3\lam%
_{0})\as_{6}\newline
\notag \\
&&+(\lam_{7}-4\lam_{0})\as_{7}+(\lam_{8}-2\lam_{0})\as_{8},
\end{eqnarray}%
and $p$ given by (\ref{palfa}).

\begin{remark}
Notice that the mass of a particle cannot be arbitrary small, since there is
a lower limit, $m\geq 1$.
\end{remark}

Hence, we extend the possible values of the 4-momentum $p:=(E,\vec{p})$ by
including those with opposite 3-momentum $\pt=(E,-\vec{p})$, as explained in
\cite{mym1}, so that%
\begin{equation}
p=(E_{p},p_{x},p_{y},p_{z})\ ,\quad E_{p}\in \nn\ ,\ p_{x},p_{y},p_{z}\in \zz%
\ ,\ p^{2}\leq 0.
\end{equation}
The algebra $\ggu$ extends the 1+1-dimensional toy model based on $\eno$
discussed in \cite{mym1}; it is defined as the algebra generated by $x_{p}^{%
\as}$ and $x_{\as+p}$, such that $p^{2}\leq 0$, satisfying the following
commutation relations:%
\bea{lcl}
\left[ x_{p_{1}}^{\alpha },x_{p_{2}}^{\beta }\right] &=&0, \\
\left[ x_{p_{1}}^{\alpha },x_{\beta +p_{2}}\right] &=&\left( \alpha ,\beta
\right) x_{\beta +p_{1}+p_{2}},  \\
\left[ x_{\alpha +p_{1}},x_{\beta +p_{2}}\right] &=&\left\{
\begin{array}{ll}
0, & \text{if~}\alpha +\beta \notin \Phi _{8}\cup \left\{ 0\right\} ; \\
\varepsilon \left( \alpha ,\beta \right) x_{\alpha +\beta +p_{1}+p_{2}}, &
\text{if~}\alpha +\beta \in \Phi _{8}; \\
-x_{p_{1}+p_{2}}^{\alpha }, & \text{if~}\alpha +\beta =0,%
\end{array}%
\right.  \label{crgu}
\eea
where $\left( \cdot ,\cdot \right) $ is the Euclidean scalar product in $%
\mathbb{R}^{8}$, the function $\varepsilon :\Phi _{8}\times \Phi
_{8}\rightarrow \left\{ -1,1\right\} $ is the asymmetry function \cite{kac,
graaf, mym1}, and
\begin{equation}\label{extra1}
x_{p}^{-\alpha }=-x_{p}^{\alpha },~~~\left[ x_{p_{1}}^{\alpha },x_{\beta
+p_{2}}\right] =-\left[ x_{\beta +p_{2}},x_{p_{1}}^{\alpha }\right],
\end{equation}%
in order to have an antisymmetric algebra.\\
Moreover, for consistency, we require that%
\begin{equation}\label{extra2}
x_{p}^{\alpha +\beta }=x_{p}^{\alpha }+x_{p}^{\beta }.
\end{equation}%
Notice that $p_{1}^{2},p_{2}^{2}\leqslant 0$ implies $\left(
p_{1}+p_{2}\right) ^{2}\leqslant 0$.

\begin{remark}
\label{ggup} Notice also that $p_{1}^{2},p_{2}^{2}<0$ imply $%
(p_{1}+p_{2})^{2}<0$. Thus, there is a subalgebra $\ggup$ of $\ggu$ with the
same commutation relations \eqref{crgu}, but with generators $x_{\as+p}$ and
$x_{p}^{\as}$ such that $p^{2}<0$ (only massive particles).
\end{remark}

\begin{proposition}
\label{e9ext} The algebra $\ggu$ with relations \eqref{crgu}, \eqref{extra1}%
, \eqref{extra2} is an infinite-dimensional Lie algebra.
\end{proposition}

The proof is in Appendix \sref{a1}.

The algebra $\ggu$ has a natural 2-grading inherited by that of $\eo$, due
to the decomposition into the subalgebra $\dotto$ and its Weyl spinor, \cite%
{mym1}. The generators $x_{\as+p}$ are fermionic (resp. bosonic) if $\as$ is
fermionic (resp. bosonic), whereas the generators $x_{p}^{\as}$ are bosonic,
due to the commutation relations \eqref{crgu}.

\section{The Lie Superalgebra $\sggu$\label{s:sggu}}

In order to turn the Lie algebra $\ggu$ into a Lie superalgebra, we exploit
the \textit{Grassmann envelope} $G(\ggu)$ of $\ggu$,%
\begin{equation}
G(\ggu):=\gguo\tp G_{0}+\gguu\tp G_{1},
\end{equation}
where $\gguo$ is the boson subalgebra of $\ggu$, $\gguu$ its fermionic part,
and $G_{0},G_{1}$ are the even, odd parts of a Grassmann algebra with
infinitely many generators. More precisely, we map each generator $X$ of $%
\ggu$ to the generator $X\tp e_{x}$ of $G(\ggu)$, where $e_{x}$ is even if $%
X $ is bosonic, odd if $X$ is fermionic, and $e_{x}\neq e_{y}$ if $X\neq Y$.
Then the graded Jacobi identity is satisfied, \cite{sh}, and one obtains, by
linearity, a Lie superalgebra, that we denote by $\sggu$.

Let us show this straightforward calculation explicitly.\newline
Let $X,Y,Z$ be generators of $\ggu$ of degree $i,j,k\in \{0,1\}$
respectively. We remind, from \cite{mym1}, that the generators $x_{p}^{\as}$
have degree 0, whereas the generators $x_{\as+p}$ have degree $\de{\as}=0$
if $\as$ is bosonic and degree $\de{\as}=1$ if $\as$ is fermionic.\newline
Let $[X,Y]$ still denote the product of $X,Y$ in $\ggu$ and $X\tp e_{x}\circ
Y\tp e_{y}$ the corresponding product in $\sggu$. Then:%
\bea{lcl}
X\tp e_{x}\circ Y\tp e_{y} &=&[X,Y]\tp e_{x}e_{y}=-[Y,X]\tp e_{x}e_{y}\\
&=&
-(-1)^{ij}[Y,X]\tp e_{y}e_{x}=-(-1)^{ij}Y\tp e_{y}\circ X\tp e_{x}
\eea
\bea{l}\label{sjac}
(-1)^{ik}((X\tp e_x \circ Y\tp e_y)\circ Z\tp e_z) + (-1)^{jk}((Z\tp e_z \circ X\tp e_x)\circ Y\tp e_y)\\
+ (-1)^{ij}((Y\tp e_y \circ Z\tp e_z)\circ X\tp e_x) = (-1)^{ik}\left[ [ X,Y ] , Z \right] \tp e_xe_ye_z \\
+(-1)^{jk} \left[ [ Z,X ] , Y \right] \tp e_ze_xe_y + (-1)^{ij}\left[ [ Y,Z ] , X \right] \tp e_ye_ze_x =\\
\left((-1)^{ik}\left[ [ X,Y ] , Z \right] +(-1)^{jk}(-1)^{k(i+j)}\left[ [ Z,X ] , Y \right]+\right.\\
\left. (-1)^{ij}(-1)^{i(j+k)}\left[ [ Y,Z ] , X \right]\right)\tp e_xe_ye_z = (-1)^{ik} J \tp e_xe_ye_z =0
\eea
where $J=0$ is the Jacobi identity for $\ggu$.

\begin{remark}
\label{r:prod} The product in the Lie superalgebra $\sggu$ is effectively
the same as in the Lie algebra $\ggu$ but its symmetry property is crucial
for the elements of the universal enveloping algebra, that appear point by
point in the model for the expanding Universe. The universal enveloping
algebra is indeed the tensor algebra $1\, \oplus\, \sggu\, \oplus\, \sggu\tp %
\sggu\, \oplus...$ modulo the relations $x\tp y - (-1)^{ij} y\tp x = x\circ
y $ for all $x,y\in \sggu$, embedded in the tensor algebra, of degree $i,j$
respectively. In particular this makes the fermions comply with the Pauli
exclusion principle: $x\tp x = 0$, for x fermionic, whereas the same
relation is trivial, $0=0$, if x is bosonic.
\end{remark}

\begin{remark}
The Lie superalgebra $\sggu$ does not involve \emph{superpartners}. The
elements are exactly the same as those of the algebra $\ggu$. The importance
we attribute to this algebra is solely due to the fulfillment of the Pauli
exclusion principle.
\end{remark}

\begin{remark}
We also notice that the use of the Grassmann envelope produces zero divisors
in the algebra whenever the same fermionic root, \emph{with the same momentum%
}, is in two interacting particles.
\end{remark}

\section{Interaction graphs\label{s:scatt}}

As mentioned in the Introduction of \cite{mym1}, the interactions have a
tree structure whose building blocks involve only three particles, and they
are expressed by the product in the underlying algebra. The scattering
amplitudes are proportional, up to normalization, to the structure constants
of the related products. An ordering of the roots has to be \textit{a priori}
set, so that the commutator between two generators is taken according to
that order. Quantum interference is obviously independent from the ordering
choice.

In this section, we set up a \textit{correspondence} between graphs and
products in the algebra $\ggu$ spanned by the generators%
\begin{equation}
\{x_{p}^{\as}\ ,\ x_{\as+p}\ ;~\as\in \Phi _{8}\ ,\ p=(E,\vec{p}),\
p^{2}\leq 0\},
\end{equation}%
and the procedure can then be trivially extended to $\sggu$ by Remark %
\sref{r:prod}.

We include the decays among the possible scatterings as \textit{resonance
interactions}, a well known and studied phenomenon in many physical
processes, as we now explain. Suppose that two
particles, one with charge $\as\in \Phi _{8}\cup \{0\}$ and momentum $p_{1}$%
, the other with charge $\bb\in \Phi _{8}\cup \{0\}$ and momentum $p_{2}$,
are present at the same space point and are such that $\as-\bb\in \Phi
_{8}\cup \{0\}$ and $E_{1}>E_{2}$; then, a decay occurs, with a certain
amplitude, producing the outgoing particles of charges $(\as-\bb)$, $\bb$
with momenta $(p_{1}-p_{2})$ and $p_{2}$ respectively, whereas the particle
with charge $\bb\in \Phi _{8}\cup \{0\}$ and momentum $p_{2}$ shifts in
space according to the expansion rule, see \cite{mym1}. The amplitude for
the decay is proportional, up to normalization, to the structure constant of
the commutator between the outgoing particles (we will comment on this
viewpoint on the decays at the end of this section).

The possible situations for an elementary interaction are depicted in the
following graphs, Figs. \ref{int1} $\div $ \ref{int3} (the \textit{resonant}
particle is also shown in case of a decay). In the graphs we use wiggly
lines for the neutral particles $x_{p}^{\as}$ and straight lines for the
charged particles $x_{\as+p}$. Red lines indicate outgoing particles and
blue lines incoming ones.

We would like to stress that the orientation of the graphs is not
significant; these are not Feynmann diagrams, although they resemble them:
only the distinction between incoming and outgoing particles matters; it
complies with 4-momentum and charge conservation.

%======= Begin Figures
\tikzset{
	>={triangle 45},
p_in/.style={thick,draw=blue, postaction={decorate},
    decoration={markings,mark=at position .5 with {\arrow[blue]{triangle 45}}}},
p_out/.style={thick,draw=magenta, postaction={decorate},
    decoration={markings,mark=at position .65 with {\arrow[magenta]{triangle 45}}}},
ph_out/.style={decorate, thick,draw=magenta, decoration={snake,amplitude=.5mm}},
ph_in/.style={decorate, thick,draw=blue, decoration={snake,amplitude=.5mm}},
 }
\begin{figure}[hbtp]
\centering%
\begin{tikzpicture}
\draw[p_in] (210:2) -- (0,0);
\node at (-.9,-1) {$x_{\beta+p_2}$};
\draw[->,ph_in] (150:2) -- (150:1);
\draw[ph_in] (150:1) -- (0,0);
\node at (-.9,1) {$x^\alpha_{p_1}$};
\draw[p_out] (0,0) -- (2,0);
\node at (1.1,-.3) {$x_{\beta+p_1+p_2}$};
\node at (0,-3) {(a)};
\end{tikzpicture}
\hspace{1cm}
\begin{tikzpicture}
\draw[p_in] (-2,0) -- (0,0);
\node at (-1,-.3) {$x_{\beta+p_1+p_2}$};
\draw[p_out] (0,0) -- (330:2);
\node at (.9,-1) {$x_{\beta+p_2}$};
\draw[->,ph_out] (0,0) -- (30:1);
\draw[ph_out] (30:1) -- (30:2);
\node at(.8,1) {$x^\alpha_{p_1}$};
\draw[p_in] (-2,-1.7) -- node[label=below:$x_{\beta+p_2}$] {} (0,-1.7);
\draw[dashed,red] (0,-1.7) -- node[label=below:$shifted$] {} (2,-1.7);
\node at (0,-3) {(b)};
\end{tikzpicture}
\caption{$[x^\as_{p_1},x_{\bb+p_2}]=(\as,\bb)x_{\bb+p_1+p_2}$; (a): $x^\as$
absorption by $x_\bb$; (b): $x^\as$ emission by $x_\bb$ (similarly for $x^%
\protect\alpha_{p_1}$ and $x_{\protect\beta+p_2}$ interchanged).}
\label{int1}
\end{figure}
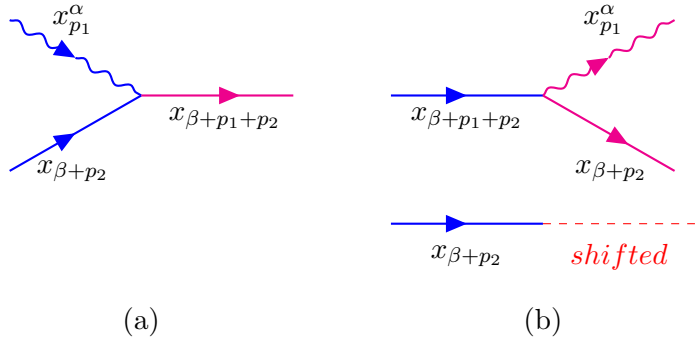

\begin{figure}[hbtp]
\centering
\begin{tikzpicture}
	\draw[p_in] (210:2) -- (0,0);
	\node at (-.9,-1) {$x_{\as+p_1}$};
	\draw[p_in] (150:2) -- (0,0);
	\node at (-.9,1) {$x_{-\as+p_2}$};
	\draw[->,ph_out] (0,0) -- (1,0);
	\draw[ph_out] (1,0) -- (2,0);
	\node at (1.1,-.36) {$x^\as_{p_1+p_2}$};
	\node at (0,-3) {(a)};
	\end{tikzpicture}
\hspace{1cm}
\begin{tikzpicture}
	\draw[->,ph_in] (-2,0) -- (-1,0);
	\draw[ph_in] (-1,0) -- (0,0);
	\node at (-1,-.36) {$x^\as_{p_1+p_2}$};
	\draw[p_out] (0,0) -- (330:2);
	\node at (.9,-1) {$x_{\as+p_1}$};
	\draw[p_out] (0,0) -- (30:2);
	\node at(.8,1) {$x_{-\alpha+p_2}$};
	\draw[p_in] (-2,-1.7) -- node[label=below:$x_{-\as+p_2}$] {} (0,-1.7);
	\draw[dashed,red] (0,-1.7) -- node[label=below:$shifted$] {} (2,-1.7);
	\node at (0,-3) {(b)};
	\end{tikzpicture}
\caption{$[x_{\as+p_1},x_{-\as+p_2}]=-x^\as_{p_1+p_2}$; (a): $x_\as$-$x_{-\as%
}$ annihilation; (b): pair creation (similarly for $x_{\protect\alpha+p_1}$
and $x_{-\protect\alpha+p_2}$ interchanged).}
\label{int2}
\end{figure}
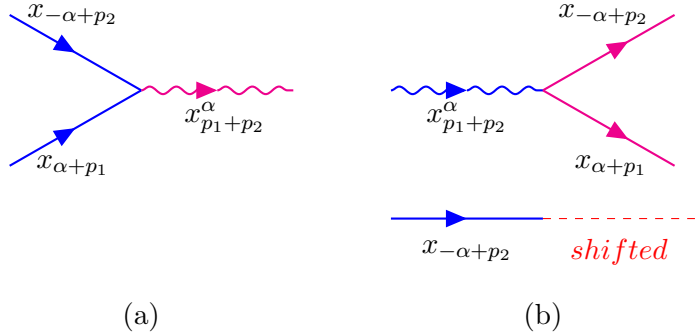

\begin{figure}[hbtp]
\centering
\begin{tikzpicture}
	\draw[p_in] (210:2) -- (0,0);
	\node at (-.9,-1) {$x_{\as+p_1}$};
	\draw[p_in] (150:2) -- (0,0);
	\node at (-.9,1) {$x_{\bb+p_2}$};
	\draw[p_out] (0,0) -- (2,0);
	\node at (1.1,-.36) {$x_{\as+\bb+p_1+p_2}$};
	\node at (0,-3) {(a)};
	\end{tikzpicture}
\hspace{1cm}
\begin{tikzpicture}
	\draw[p_in] (-2,0) -- (0,0);
	\node at (-1,-.36) {$x_{\as+\bb+p_1+p_2}$};
	\draw[p_out] (0,0) -- (330:2);
	\node at (.9,-1) {$x_{\as+p_1}$};
	\draw[p_out] (0,0) -- (30:2);
	\node at(.8,1) {$x_{\bb+p_2}$};
	\draw[p_in] (-2,-1.7) -- node[label=below:$x_{\bb+p_2}$] {} (0,-1.7);
	\draw[dashed,red] (0,-1.7) -- node[label=below:$shifted$] {} (2,-1.7);
	\node at (0,-3) {(b)};
	\end{tikzpicture}
\caption{$[x_{\as+p_1},x_{\bb+p_2}]=\eps(\as,\bb) x_{\as+\bb+p_1+p_2}$; (a):
$x_\as$-$x_{\bb}$ scattering; (b): $x_{\as+\bb}$ decay into $x_\as$ and $x_
\bb$ (similarly for $x_{\protect\alpha+p_1}$ and $x_{\bb+p_2}$
interchanged). }
\label{int3}
\end{figure}
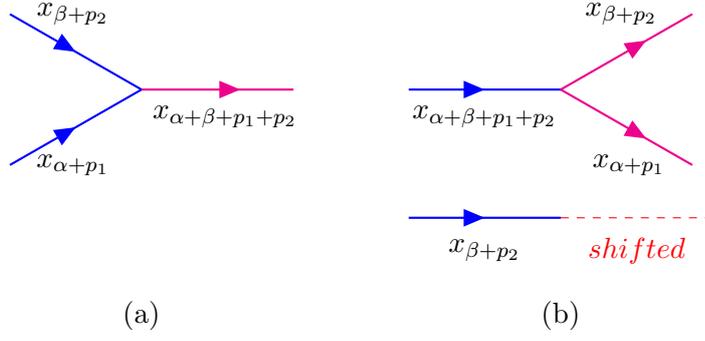

A particular case represented by Fig. \ref{int3} is the interaction among
gluons.

Notice that for each interaction as in (a) of Figs. \ref{int1} $\div $ \ref%
{int3} there is an amplitude for a shift of the two particles without
interaction. This allows for an interaction as in (b) of the same Fig. at a
later time. %======= End Figures

\section{The Poincar\'e group}

\label{s:Poin}

We refer to section \textbf{2.3} of our previous paper \cite{mym1}, in
particular we denote by $\rho_1,\rho_2$ the roots $k_5-k_6$ and $k_5+k_6$
respectively.

We have a complex $\au\oplus\au$ Lie algebra $\lore$ generated by $x_{\pm
\rho_1}$, $x_{\pm \rho_2}$ and the corresponding Cartan generators $%
h_{\rho_1}$, $h_{\rho_2}$.

The spin subalgebra $\spin\in\lore$ is the compact form of the subalgebra
with generators $R^+:=x_{\rho_1}+x_{\rho_2}$, $R^-:=x_{-\rho_1}+x_{-\rho_2}$
and $H_R:=\frac12(h_{\rho_1}+h_{\rho_2})$, namely $\spin$ is generated by $%
R^+ + R^-$, $i\left( R^+ - R^-\right)$ and $iH_R$.

We denote by $w$ the Pauli-Lubanski vector and we classify the $\ggu$
generators $x_{\as+p}$ or $x^\as_p$ with respect to $m^2=-p^2$ and $w^2$,
the two Casimir invariants of the Poincar\'e group. We use the shorthand
notation $k:=\pm k_1\pm k_2\pm k_3\pm k_4\pm k_7\pm k_8$, $k_e$ (resp. $k_o$%
) when $k$ has an even (resp. odd) number of $+$ signs.
\bea{ll} (-p^2,w^2)
& \text{generator}\\
\hline \vspace{-1em}\\
(m^2,0) & x_{\as+p}, x^\as_p \, |\, \as=\pm k_i \pm k_j, \ i,j\notin
\{5,6\}, \, -p^2=m^2\vspace{2pt}\\
(m^2, \frac34 m^2) & x_{\as+p}, x^\as_p \, |\, \as=\frac12(k_o\pm
(k_5-k_6)), \, -p^2=m^2\vspace{2pt}\\
(m^2, \frac34 m^2) & x_{\as+p}, x^\as_p \, |\, \as=\frac12(k_e\pm
(k_5+k_6)), \, -p^2=m^2\vspace{2pt}\\
(m^2, 2 m^2) & x_{\as+p}, x^\as_p \, |\, \as=k_i\pm k_5 \text{ or } \as%
=k_i\pm k_6, \, -p^2=m^2\vspace{2pt}\\
(0,0) & \text{ all generators $x_{\as+p}$ or $x^\as_p$ such that $p^2=0$}
\eea

Let $\ggue$ be the extension of $\ggu$ that includes all timelike and
lightlike momenta (not necessarily with integer component), and let $\ggupe$
the subalgebra $\ggue$ of massive particles, namely the extension of the
subalgebra $\ggup$ introduced in Remark \sref{ggup}. We regard the following
proposition as fundamental for the relativistic behavior of our model.

\begin{proposition}
There is a natural action of the Poincar\`e group $\poig\, :\, \ggue \to %
\ggue$. Let $\poi=(\Lambda,a)$ be an element of $\poig$, where $\Lambda$ is
a Lorentz transformation and $a$ a translation.\newline
The action extends by linearity the following action on the generators $x_{%
\as+p}$ and $x^\as_p$ of $\ggue$.\newline
\ben

\item If $p^2<0$, fix a transformation $\Lambda_p$ such that $\Lambda_p
(m,0,0,0) = p$ and let $W(\Lambda,p):=\Lambda^{-1}_{\Lambda p} \Lambda
\Lambda_p$ be the Wigner rotation induced by $\Lambda$ \be\label{poi1} \poi%
(x_{\as+p}) = e^{ia\cdot \Lambda p} e^{\adj(R)} x_{\as+\Lambda p}\ , \poi(x^%
\as_p) =\ e^{ia\cdot \Lambda p} e^{\adj(R)} x^\as_{\Lambda p} \ee
where $\adj(R)$ is the adjoint action of the generator $R\in\spin$ of the
Wigner rotation $W(\Lambda,p)$.

\item If $p^2=0$ and $w^2=0$ the action reduces to \be\label{poi2} \poi(x_{%
\as+p}) = e^{ia\cdot \Lambda p} e^{i\theta(\Lambda) \lam} x_{\as+\Lambda p}\
, \ \poi(x^\as_p) = e^{ia\cdot \Lambda p} e^{i\theta(\Lambda) \lam} x^\as%
_{\Lambda p} \ee
where $\lam= 0, \pm \frac12, \pm 1$ is the helicity of $\as$ and $\theta$ is
the angle of the $SO(2)$ rotation along the direction of $\vec p$, analogous
to the Wigner rotation of the massive case. \een
The Poincar\`e group is a subgroup of the automorphism group of $\ggupe$.
\end{proposition}

\begin{proof}
	The action $\poi$ on each generator with a certain mass and spin/helicity acts as the irreducible \emph{induced representation}, introduced by Wigner, \cite{wigner}.\\
	We only need to prove that it is an automorphism of $\ggupe$, namely that $\poi$ is non-singular and preserves the Lie product \eqref{crgu}. Part of the proof is similar to the classical one, see Lemma 4.3.1 in \cite{carter}.\\
	The fact that $\poi$ is non-singular comes from the obvious existence of its inverse transformation. We are left with the proof that $\poi([X,Y])=[\poi(X),\poi(Y)]$.\\
	Let us consider in particular $\poi(x_{\as+p})$ in \eqref{poi1}. Since $R$ is an $\eo$ generator then $\adj(R)$ is nilpotent, namely
	 $\der^r=0$ for some $r$ and
	\be e^{\der} = 1+\der+\frac{\der^2}{2!}+ ... +\frac{\der^{r-1}}{(r-1)!}\ee
	We have
	\bea{ll}
	\dfrac1{s!}\der^s [x,y] &=
	\dfrac1{s!}\sum_{i=0}^s{\binom{s}{i} [\der^i x,\der^{s-i}y]}\\
	&=\sum_{\substack{i,j\\i+j=s}}{\dfrac1{i!\, j!}\left[\der^i x,\der^jy\right]}
	\eea
	and also that $\der^t=0$ for $t\ge r$ implies
	\be
	\sum_{\substack{i,j}}{\dfrac1{i!\, j!}\left[\der^i x,\der^j y\right]}=0 \text{ if $i+j\ge r$}
	\ee
	Let $\as+\bb\in\Phi_8$ and $p_1^2,p_2^2<0$. We get:
	\bea{l}
	\poi([x_{\as+p_1},x_{\bb+p_2}]) = \poi(\ep(\as,\bb)x_{\as+\bb+ p_1+p_2}) = e^{ia\cdot \Lambda (p_1+p_2)} e^{\adj(R)} [x_{\as+\Lambda p_1},x_{\bb+\Lambda p_2}]\\
	\hspace{4em}=\ e^{ia\cdot \Lambda (p_1+p_2)}\sum_{s\ge0}\sum_{\substack{i,j\\i+j=s}} {\dfrac1{i!\, j!}\left[\der^i x_{\as+\Lambda p_1},\der^j x_{\bb+\Lambda p_2}\right]}\\
	\hspace{4em}=\ e^{ia\cdot \Lambda p_1+\Lambda p_2}\sum_{i\ge0}\sum_{j\ge0}{\dfrac1{i!\, j!}\left[\der^i x_{\as+\Lambda p_1},\der^j x_{\bb+\Lambda p_2}\right]}\\
	\hspace{4em}=\ [\poi (x_{\as+p_1}),\poi (x_{\bb+p_2})]
	\eea
	Similarly for the other commutators in \eqref{crgu}.
\end{proof}

The action $\poi$ can be easily extended to $\sggu$ by acting accordingly on
the Grassmann variable in order to get the variable associated to the
transformed generators of $\ggu$.

\section{Initial Quantum State\label{s:stat}}

The initial quantum state of our model of the expanding early Universe is an
element of the universal enveloping algebra $\usgu$ of $\sggu$, namely an
element of the tensor algebra built on the generators of $\sggu$ \textit{%
modulo} the relations defining the product in the algebra itself. The
initial generators are all in pairs with opposite helicity and opposite
3-momentum, \cite{mym1}, and have a phase or amplitude associated to each of
them as a complex coefficient. The interactions and expansions starting from
the initial state are such that \textit{locally} the quantum state is an
element of the universal enveloping algebra. Interference plays the crucial
role in the quantum behavior of the model, including repulsive \textit{versus%
} attractive forces. \textit{The quantum nature of gravity appears through
the quantum nature of spacetime}: at every cosmological instant, a point in
space has an amplitude which is the sum of the amplitudes for particles to
be at that point.

The initial state has the mean energy of the Universe concentrated on the
generators that interact with each other at $t=0$. The choice of the initial
state is crucial in determining the likelihood for the existence of
particles and of an eventual symmetry breaking. It is beyond the scope of
this paper to investigate this subject in depth; an algorithm based on the
algebra and the expansion rule that we have introduced can be the basis for
computer calculations, which should shed some light on the physical
consequences of the choice of the initial quantum state.

\section{Vertex-type algebra and \textit{Gravitahedra}\label{s:vert}}

Space expansion leads to an enrichment of the algebra. The locality of
interactions suggests to embed the algebra in a vertex-type operator
algebra, in which the generators of $\sggu$ act as vertex operators on a
discrete space that is being built up, step by step, by $\sggu$ driven
interactions.

The tree structure of the interactions allows for a description of
scattering amplitudes in terms of \textit{associahedra} or \textit{%
permutahedra}, \cite{marni}-\nocite{nima1}\nocite{miz}\nocite{nima2}\nocite%
{stash1}\nocite{stash2}\nocite{tonks}\cite{loday}, with structure constants
attached to each vertex; see Fig. \ref{f:Ktrees} for the interaction of four
particles, producing the \textit{associahedron} $\Kass_{4}$. A vertex is
interpreted as an interaction with universal time flowing from top to bottom
in the trees of Fig. \ref{f:Ktrees}. However, if one includes the \textit{%
gravitational effect} of space expansion, one should describe the
interactions through permutahedra $\Perm_{n-1}$ rather than associahedra;
see Fig. \ref{f:Ptrees} for the interaction of four particles, producing the
\textit{permutahedron} $\Perm_{3}$.
\begin{figure}[hbtp]
\centerline{		\includegraphics[scale=0.2]{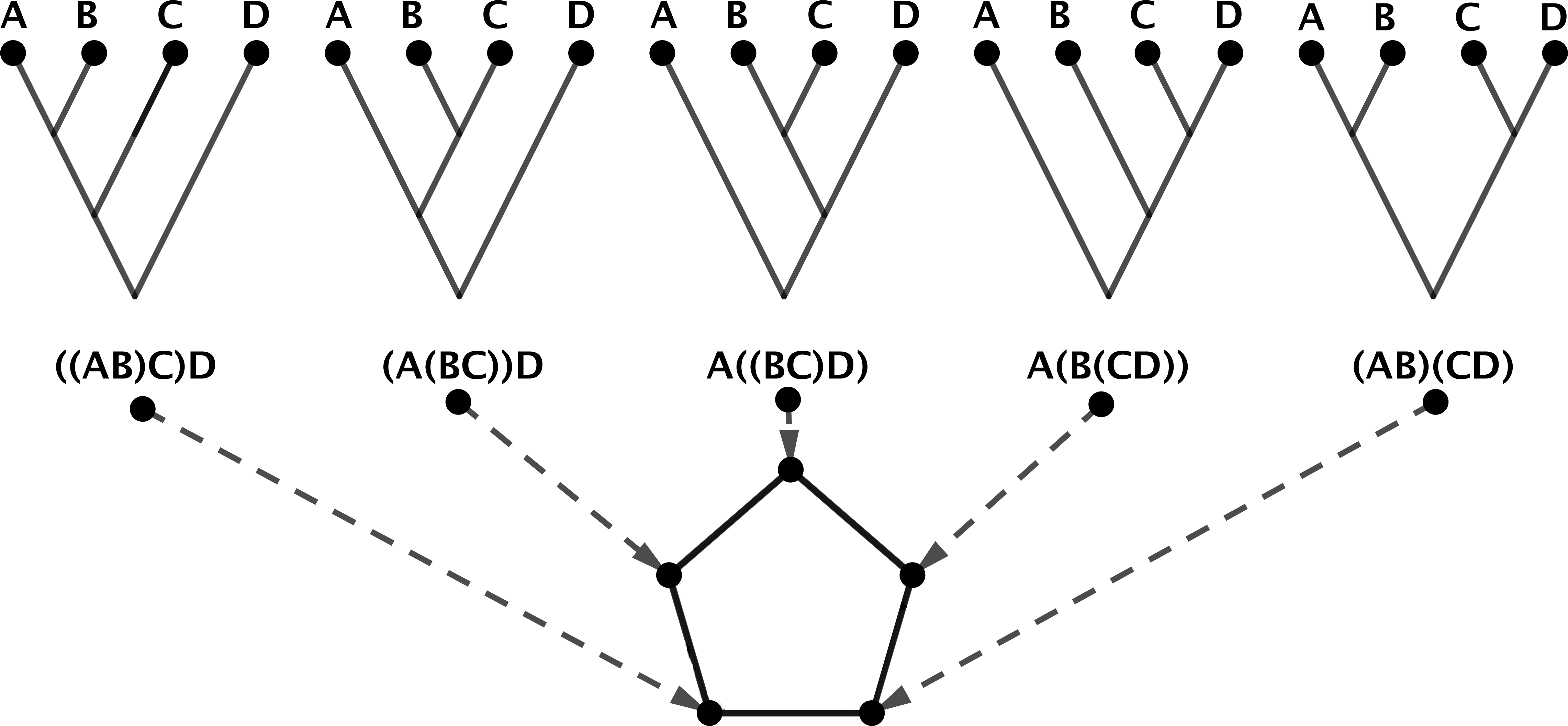}}
\caption{Associahedron $\Kass_4$. Adjacent vertices $(st)u \to s(tu)$, for
sub-words $s,t,u$}
\label{f:Ktrees}
\end{figure}
\begin{figure}[hbtp]
\centerline{		\includegraphics[scale=.9]{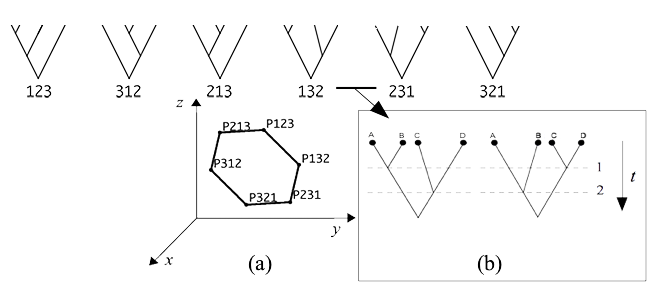}}
\caption{Permutahedron $\Perm_3$. Interaction of 4 particles.}
\label{f:Ptrees}
\end{figure}
The two trees in Fig. \ref{f:Ptrees} (b) are different due to the spreading
of particles in space, because the same interactions occur at different
times (represented by the horizontal lines).

A complete graphical description of the interactions, including the
spacetime effects, hence gravity, can be quite complicated and needs a deep
study. A research program with this goal has initiated, and the name \textit{%
gravitahedra} has been coined for the polytopes that will eventually, and
hopefully, describe such interactions.

The fact that \textit{locally} the quantum state is an element of the
universal enveloping algebra means that we can assign to it labels $q$ of
space $Q$, which are triples of rational numbers, due to the expansion by $%
\vec{p}/E$, where $E,p_{x},p_{y},p_{z}$ are integers, \cite{mym1}. The
vertex-type algebra is therefore the algebra $\usgu(Q)$, whose relations
have been extended in order to include the commutation of elements with
different space-labels.

\section{Conclusion}

In the pair of papers given by \cite{mym1} and the present paper, we have
presented an intrinsically quantum and relativistic theory of the creation
of spacetime starting from a quantum state as cosmological boundary
condition, which we conceive to play a key role in any fundamental theory of
Quantum Gravity. We have discussed the general framework of a workable
model, based on a rank-12 infinite dimensional Lie superalgebra, which can
be applied to the quantum era of the first cosmic evolution. Our model can
accommodate the degrees of freedom of the particles we know, \textit{without
superpartners}, namely spin-$\frac{1}{2}$ fermions and spin-$0$ and spin-$1$
bosons obeying the proper statistics.

The quantum nature of gravity is intrinsically unobservable, because
observation implies the destruction of the entanglement and the collapse of
the wavefunction. Thus, the deal in Quantum Gravity is the following : the
intrinsically quantum and relativistic description of an intrinsically
unobservable regime should be made consistent with the existence of a
macroscopic observer, and thus of a (semi)classical observational symmetry,
emerging in the theormodynamical/macroscopic limit in which the entanglement
becomes irrelevant. Our model tackles this crucial issue of Quantum Gravity,
and solves it with elegance : indeed, the Poincar\'{e} group emerges from
both the \textquotedblleft spin" sector ($\eo$) and the kinematical sector
(complementary of $\eo$ in $\ggu$) of the Lie superalgebra $\sggu$. Besides
the absence of superpartners and the implementation of the Pauli exclusion
principle, the emergence of the Poincar\'{e} group is a crucial feature of
our model. We should stress that, of course, the Poincar\'{e} group can be
defined only in the thermodynamical limit in which the observer can be
consistently decoupled from the evolutive dynamics of the Universe, given
\textit{in toto} by $\sggu$. Especially in an early Universe, the
back-reaction of the observer on the object of the observation should be
relevant, and thus the abstraction of a decoupled and distinct observer is
not totally consistent during the early stages of the Universe.

Many physical properties have still to be verified and/or fulfilled, like
the proton decay, the confinement of quarks, the attractive nature of
gravity on the large scale. The general framework of the model leaves
however a great freedom of choice, and this is as a benefit for those who
believe this is a promising approach and wish to explore it.\newline

There is much left for future work, to start with the definition of a
particular quantum initial state allowing to perform some preliminary
computer calculations that may give an idea of how the model effectively
works. In particular, the density matrix, von Neumann entropy, mean energy,
scattering amplitudes can be \textit{explicitly calculated} according to our
model.

We end this series of two papers by recapitulating what we consider the main
physical features of our approach:

\begin{description}
\item[I)] spacetime is the outcome of the interactions driven by an
infinite-dimensional Lie superalgebra $\sggu$; it is discrete, finite and
expanding;

\item[II)] the algebra $\sggu$ incorporates 4-momentum and charge
conservation; it involves fermions and bosons, with fermions fulfilling the
Pauli exclusion principle;

\item[III)] $\sggu$ is a Lie superalgebra without any supersymmetry forcing
the existence of superpartners for the particles of the Standard Model;

\item[IV)] every particle has positive energy and it is either timelike or
lightlike;

\item[V)] the initial state is an element of the universal enveloping
algebra of $\sggu$;

\item[VI)] the interactions are local, and the whole algebraic structure is
a vertex-type algebra, due to a mechanism for the expansion of space (in
fact, an expansion of matter and radiation);

\item[VII)] the emerging spacetime inherits the quantum nature of the
interactions, hence Quantum Gravity is an expression for quantum spacetime -
in particular, there is no spin-2 particle;

\item[VIII)] the Poincar\'{e} group has a natural action on the local
algebra;

\item[IX)] once an initial state is fixed, the model can be viewed as an
algorithm for explicit computer calculations of physical quantities, like
scattering amplitudes, density matrix, partition function, mean energy, von
Neumann entropy, e\textit{tc.}.
\end{description}

\appendix

\section{Appendix\label{a1}}

We prove Proposition \sref{e9ext}.

The algebra $\ggu$ with relations \eqref{crgu}, \eqref{extra1}, %
\eqref{extra2} is obviously infinite dimensional, and its product is
antisymmetric. We only need to prove that it fulfills the Jacobi identity.%
\newline

Throughout the proof we strongly rely on the following standard results, see
\cite{mym1} and Refs. therein.

\begin{proposition}
\label{sproots} For each $\as,\bb \in \Phi_8$ the scalar product $(\as,\bb)
\in \{\pm 2, \pm 1, 0\}$; $\as + \bb$ ( respectively $\as - \bb$) is a root
if and only if $(\as , \bb) = -1$ (respectively $+1$); if both $\as+\bb$ and
$\as-\bb$ are not in $\Phi_8\cup\{0\}$ then $(\as,\bb)=0$.\newline
For $\as , \bb \in \Phi_8$ if $\as+\bb$ is a root then $\as - \bb$ is not a
root.
\end{proposition}

\begin{proposition}
\label{epsprop} The asymmetry function $\ep$ satisfies, for $\as , \bb, \gh %
\in \Ll$:
\bes\ba{rrcl} i) & \ep (\as + \bb, \gh) & =& \ep (\as , \gh)\ep (%
\bb , \gh) \\
ii) &\ep (\as ,\bb+\gh) & = &\ep (\as , \bb)\ep (\as , \gh) \\
iii) &\ep (\as , \as) & = & (-1)^{\frac12 (\as,\as)} \impl \ep (\as , \as) =
-1 \text{ if } \as\in\Phi_8\\
iv) &\ep (\as , \bb) \ep (\bb , \as) &=& (-1)^{(\as , \bb)} \impl \ep (\as , %
\bb) = -\ep (\bb , \as)\text{ if } \as,\bb,\as+\bb\in\Phi_8\\
v) &\ep (0 , \bb) &=& \ep (\as , 0) = 1 \\
vi) &\ep (-\as , \bb) &=& \ep (\as , \bb)^{-1} = \ep (\as , \bb) \\
vii) &\ep (\as , -\bb) &=& \ep (\as , \bb)^{-1} = \ep (\as , \bb) \\
\ea\ees
\end{proposition}

By linearity it is sufficient to prove that the Jacobi identity holds for
the generators of the algebra. For each triple of generators $X,Y,Z$ we
write \be J_{1}:=\left[ [X,Y],Z\right] \ , \qquad J_{2}:= \left[ [Z,X],Y%
\right] \ , \qquad J_{3}:= \left[ [Y,Z],X\right] \ee We want to prove that $%
J:=J_{1}+J_{2}+J_{3}=0$.\\
For $p\neq 0$ we call the generators $x_{p}^{\as}$ of \textit{type 0} and $%
x_{\as+p}$ of \textit{type 1}.\\
We consider the various cases.\\

\begin{description}
\item[a)] \textit{At least} one of $X,Y,Z$ is of type-0

\begin{description}
\item[a1)] If $X,Y,Z$ are all of the type-0 then Jacobi holds trivially.

\item[a2)] If $X=x_{p_{1}}^{\as},Y=x_{p_{2}}^{\bb}$ are of type 0 and $Z=x_{%
\gh+p_{3}}$ is of type 1 then $J_{1}=0$, $J_{2}=(\as,\gh)(\bb,\gh)x_{\gh%
+p_{1}+p_{2}+p_{3}}$ and $J_{3}=-(\as,\gh)(\bb,\gh)x_{\gh+p_{1}+p_{2}+p_{3}}$%
, hence $J=0$.

\item[a3)] If $X=x_{p_{1}}^{\as}$ is of type 0 and $Y=x_{\bb+p_{2}},Z=x_{\gh%
+p_{3}}$ are of type 1, then $J_{1}=(\as,\bb)[x_{\bb+p_{1}+p_{2}},x_{\gh%
+p_{3}}]$, $J_{2}=(\as,\gh)[x_{\bb+p_{2}},x_{\gh+p_{1}+p_{3}}]$ and $%
J_{3}=-[x_{p_{1}}^{\as},[x_{\bb+p_{2}},x_{\gh+p_{3}}]]$. We have 3 cases:

\begin{description}
\item[a3.i)] $\bb+\gh\not\in \Phi _{8}\cup \{0\}$ then $J_{1}=J_{2}=J_{3}=0$;

\item[a3.ii)] $\bb+\gh\in \Phi _{8}$ then $J_{3}=-(\as,\bb+\gh)\ep(\bb,\gh%
)x_{\bb+\gh+p_{1}+p_{2}+p_{3}}=-(J_{1}+J_{2})$;

\item[a3.iii)] $\bb+\gh=0$ then $J_{1}=-(\as,\bb)x_{p_{1}+p_{2}+p_{3}}^{\bb}$%
, $J_{2}=(\as,\bb)x_{p_{1}+p_{2}+p_{3}}^{\bb}$ and $J_{3}=0$, hence $J=0$.
\end{description}
\end{description}

\item[b)] None of $X,Y,Z$ is of type-0. Let $X=x_{\as+p_{1}},Y=x_{\bb%
+p_{2}},Z=x_{\gh+p_{3}}$ be all of type 1. For any two roots of $\Phi _{8}$,
say $\as,\bb$ without loss of generality, we have three cases:

\begin{description}
\item[b1)] $\as+\bb\not\in \Phi _{8}\cup \{0\}$:\newline

\begin{description}
\item[b1.i)] if both $\as+\gh,\bb+\gh\not\in \Phi _{8}\cup \{0\}$ then $J=0$
trivially;

\item[b1.ii)] if $\bb+\gh\not\in \Phi _{8}\cup \{0\}$ and $\as+\gh\in \Phi
_{8}\cup \{0\}$ then $J_{1}=J_{3}=0$. Since both $(\as,\bb),(\bb,\gh)\in
\{0,1,2\}$ then $(\as+\gh,\bb)\geq 0$ hence if $\as+\gh\in \Phi _{8}$, then $%
\as+\bb+\gh\not\in \Phi _{8}\cup \{0\}$ and $J_{2}=0$. On the other hand if $%
\as=-\gh$ then $J_{2}=[x_{p_{1}+p_{3}}^{\as},x_{\bb+p_{2}}]=(\as,\bb)x_{\bb%
+p_{1}+p_{2}+p_{3}}$. But $(\bb,\gh)=-(\bb,\as)$ and $(\as,\bb),(\bb,\gh)\in
\{0,1,2\}$ imply $(\as,\bb)=0$ hence $J=0$;

\item[b1.iii)] if $\bb+\gh\in \Phi _{8}$ and $\as+\gh\in \Phi _{8}$ then $%
J_{2}=\ep(\gh,\as)[x_{\as+\gh+p_{1}+p_{3}},x_{\bb+p_{2}}]$ and $J_{3}=\ep(\bb%
,\gh)[x_{\bb+\gh+p_{2}+p_{3}},x_{\as+p_{1}}]$. If $\as+\bb+\gh\not\in \Phi
_{8}\cup \{0\}$ then $J_{2}=J_{3}=0$ hence $J=0$. If $\as+\bb+\gh\in \Phi
_{8}$ then $J_{2}+J_{3}=\ep(\gh,\as)(\ep(\gh,\bb)\ep(\as,\bb)+\ep(\bb,\gh)\ep%
(\bb,\as))x_{\as+\bb+\gh+p_{1}+p_{2}+p_{3}}$. Since $2=(\as+\bb+\gh,\as+\bb+%
\gh)=6+2(\as,\bb)+2(\bb,\gh)+2(\as,\gh)=2+2(\as,\bb)$, we get $(\as,\bb)=0$
and, from Proposition \sref{epsprop}, $\ep(\as,\bb)=\ep(\bb,\as)$ and $\ep(%
\gh,\bb)=-\ep(\bb,\gh)$, implying $J_{2}+J_{3}=0$ and $J=0$. Finally if $\as+%
\bb+\gh=0$ then $(\as,\bb)=(\as,-\as-\gh)=-2+1=-1$ and $\as+\bb$ would be a
root, contradicting the hypothesis.

\item[b1.iv)] if $\bb+\gh\in \Phi _{8}$ and $\as+\gh=0$ then $J_{2}=(\as,\bb%
)x_{\bb+p_{1}+p_{2}+p_{3}}$ and $J_{3}=\ep(\bb,\as)\ep(\bb-\as,\as)x_{\bb%
+p_{1}+p_{2}+p_{3}}=-x_{\bb+p_{1}+p_{2}+p_{3}}$. But $(\as,\bb)=-(\gh,\bb)=1$
hence $J_{2}+J_{3}=0$ and $J=0$.

\item[b1.v)] If $\bb+\gh\in \Phi _{8}$ and $\as+\gh\not\in \Phi _{8}\cup
\{0\}$ then $J_{2}=0$ and $(\bb,\as)\geq 0$, $(\gh,\as)\geq 0$ imply $(\bb+%
\gh,\as)\geq 0$ hence $\bb+\gh+\as\not\in \Phi _{8}\cup \{0\}$ therefore $%
J_{3}=\ep(\bb,\gh)[x_{\bb+\gh+p_{2}+p_{3}},x_{\as+p_{1}}]=0$ and $J=0$.

\item[b1.vi)] If $\bb+\gh=0$ and $\as+\gh\in \Phi _{8}$ then\\
$J_{2}=-\ep(\as,\bb)\ep(\as-\bb,\bb)x_{\as+p_{1}+p_{2}+p_{3}} =x_{\as
+p_{1}+p_{2}+p_{3}}$ and\\
 $J_{3}=-[x_{p_{2}+p_{3}}^{\bb},x_{\as+p_{1}}]=-x_{
\as+p_{1}+p_{2}+p_{3}}$, being $(\as,\bb)=-(\as,\gh)=1$, implying $J=0$.

\item[b1.vii)] If $\bb+\gh=0$ and $\as+\gh=0$ then $J_{2}=[x_{p_{1}+p_{3}}^{%
\bb},x_{\bb+p_{2}}]=2x_{\bb+p_{1}+p_{2}+p_{3}}$ and $%
J_{3}=-[x_{p_{2}+p_{3}}^{\bb},x_{\bb+p_{1}}]=-2x_{\bb+p_{1}+p_{2}+p_{3}}$
and $J=0$.

\item[b1.viii)] If $\bb+\gh=0$ and $\as+\gh\not\in \Phi _{8}\cup \{0\}$ then
$J_{2}=0$; $(\as,\bb)=0$ since $(\as,\bb)\geq 0$ and $(\as,\gh)=-(\as,\bb%
)\geq 0$, therefore $J_{3}=-[x_{p_{2}+p_{3}}^{\bb},x_{\as+p_{1}}]=0$ and $%
J=0 $.
\end{description}

\item From now on $\as+\bb,\as+\gh,\bb+\gh\in \Phi _{8}\cup \{0\}$.

\item[b2)] $\as+\bb\in \Phi _{8}$:

\begin{description}
\item[b2.i)] If $\as+\gh,\bb+\gh\in \Phi _{8}$ then $(\as+\bb+\gh,\as+\bb+\gh%
)=0$ hence $\as+\bb+\gh=0$. Then $J_{1}=-\ep(\as,\bb)x_{p_{1}+p_{2}+p_{3}}^{%
\as+\bb}$, $J_{2}=\ep(\as+\bb,\as)x_{p_{1}+p_{2}+p_{3}}^{\bb}$, $J_{3}=\ep(%
\bb,\as+\bb)x_{p_{1}+p_{2}+p_{3}}^{\as}$. Since $\ep(\as+\bb,\as)=\ep(\bb,\as%
+\bb)=\ep(\as,\bb)$, $x_{p}^{\as+\bb}=x_{p}^{\as}+x_{p}^{\bb}$, see %
\eqref{extra2}, we get $J=0$.

\item[b2.ii)] If $\as+\gh\in \Phi _{8}$ and $\bb+\gh=0$ then $\as-\bb\in
\Phi _{8}$ which is impossible.

\item[b2.iii)] If $\as+\gh=0$ and $\bb+\gh\in \Phi _{8}$ then $\bb-\as\in
\Phi _{8}$ which is impossible.

\item[b2.iv)] If $\as+\gh=0$ and $\bb+\gh=0$ then $\as=\bb$ which is
impossible.
\end{description}

\item[b3)] $\as+\bb=0$:

\begin{description}
\item[b3.i)] If $\as+\gh\in \Phi _{8}$ then $\bb+\gh=-\as+\gh\not\in \Phi
_{8}$; we can only have $\bb+\gh=0$ implying $\as=\gh$, that contradicts $\as%
+\gh\in \Phi _{8}$.

\item[b3.ii)] If $\bb+\gh\in \Phi _{8}$ then $\as+\gh=-\bb+\gh\not\in \Phi
_{8}$; we can only have $\as+\gh=0$ implying $-\as=\bb=\gh$ that contradicts
$\bb+\gh\in \Phi _{8}$.

\item[b3.iii)] If both $\as+\gh=0$ and $\bb+\gh=0$ then $\as=\bb$ which
contradicts $\as+\bb=0$.
\end{description}
\end{description}
\end{description}

This ends the proof. \esquare%================================

\end{document}